\documentclass{article}
\usepackage{amsmath,amssymb,bbm}
\usepackage{amsthm}
\usepackage{graphicx}
\usepackage{subfigure}

\newtheorem{theorem}{Theorem}
\newtheorem{lemma}[theorem]{Lemma}

\newtheorem{definition}[theorem]{Definition}
\newtheorem{assumption}[theorem]{Assumption}

\newtheorem{remark}[theorem]{Remark}
\makeatletter
\newcommand*{\rom}[1]{\expandafter\@slowromancap\romannumeral #1@}
\makeatother

\begin{document}

\title{
   Optimal Liquidation in a Finite Time Regime Switching Model with Permanent and Temporary Pricing Impact
    }
\author{Baojun Bian\thanks{Department of Mathematics, Tongji University,
Shanghai 200092, China. {\tt bianbj@tongji.edu.cn}. This work was supported in part
by National Science Foundation of China(No.11371280,
No.71090404). }\; and
Nan Wu\thanks{Department of Mathematics, Imperial College, London SW7 2BZ, UK. {\tt nan.wu07@imperial.ac.uk}}   and
Harry Zheng\thanks{Department of Mathematics, Imperial College, London SW7 2BZ, UK.
{\tt h.zheng@imperial.ac.uk}.}}
\date{}
\maketitle

\noindent{\bf Abstract}.
In this paper we discuss the  optimal liquidation over a finite time horizon until the exit time.  The drift and diffusion terms of the asset price are general functions depending on  all variables including control and
market regime. There is also a local nonlinear transaction cost associated to the liquidation. The model deals with both the permanent impact and the temporary impact in a regime switching framework. The problem can be solved with the dynamic programming principle.  The optimal value function is the unique continuous viscosity solution to  the HJB equation and can be computed with the finite difference method.
\bigskip

\bigskip

\section{Introduction} \label{S:intro}
Optimal liquidation has attracted active research in recent years due to the liquidity risk. In a frictionless and competitive market an asset can be traded with any amount at any rate without affecting the market price of the asset. The optimal liquidation then becomes an optimal stopping problem which maximizes the expected liquidation value at the optimal stopping time.
In an incomplete market with trading constraints on the volume and the rate and with the liquidation impact on the underlying asset price, the optimal liquidation is difficult to model and to solve.

Despite the wide recognition of the importance of the liquidity risk,
there is no universal agreement on the definition of liquidity.
In the academic literature the liquidity is usually defined in terms of
the bid-ask spread and/or the transaction cost
whereas in the practitioner literature the illiquidity is often
viewed as the inability of buying and selling securities.
  Black \cite{Black71} classifies the following four major properties of
the liquidity:
the immediacy of the transaction, the tightness of the spread, the resiliency
of the market,
and the depth of the market. The concept of liquidity can be summarized as
the ability for traders to execute large trades rapidly at a price close
to current market price. The liquidity risk refers to the loss stemming
from the cost of liquidating a position.

Due to lack of universal agreement on the definition of liquidity, there are many different forms of mathematical characterizations. Apart from commonly used transaction cost and bid-ask spread and trading constraints (Cvitanic and Karatzas \cite{CK96}, Jouini \cite{J00}, etc.), the other descriptions include, for example, that the order of a large investor adversely affects the stock price before being exercised (Bank and Baum \cite{BB04}),  that the market has a supply curve that depends on the order size of investors (\c{C}etin et al. \cite{CJP04}), that  trading can only happen at jump times of a Cox process  (Gassiat et al.~\cite{GGP11}), that the asset price is affected by the permanent   and temporary impact of liquidation (Schied and Sch\"{o}neborn \cite{SS07}), etc.
Once the mathematical framework is chosen for the liquidity risk one can then study specific  problems such as the arbitrage pricing theory, the optimal investment and consumption,  etc., see  \cite{BB04, CJP04, CK96, GGP11, J00, SS07} and references therein.

This paper studies the optimal liquidation in the presence of liquidity risk. There are several variations in the problem formulation in the literature, including finite or infinite time horizon, continuous trading or optimal stopping,  geometric Brownian motion (GBM) asset price process or Markov modulated process, etc.
Pemy et al.~\cite{PZY07} study the optimal liquidation  over an infinite time horizon. The stock
price follows a GBM process with an extra term that reflects the permanent impact of liquidation
on the asset price and there is no temporary impact.  It is a constrained control problem which implicitly
assumes that the stock holdings will never be sold out for any admissible trading strategies.
The value function is
the unique continuous viscosity solution to the Hamilton-Jacobi-Bellman (HJB) equation (two state variables and no time variable).
In the continuous time finite state Markov chain framework
Pemy and Zhang \cite{PZ06} study an optimal stopping problem of liquidation in finite time horizon.
Pemy et al.~\cite{PZY08} discuss the optimal liquidation over an infinite time, similar to that in \cite{PZY07}. The main difference is that the asset price follows a
GBM process in which the drift and diffusion coefficients are determined by market regimes and  the temporary impact of  liquidation is reflected in the payoff function and there is no permanent impact.  The assumptions and the conclusions are basically the same as those in \cite{PZY07}.

In this paper we discuss the  optimal liquidation over a finite time horizon until the exit time.  The drift  and diffusion  coefficients $\mu$ and $\sigma$ of the asset price   are general functions depending on  all variables including control (see (\ref{E:ModelSetup4})), which implies the trading may cause the permanent impact on the asset price. There are also  nonlinear transaction costs associated to the trading through the temporary pricing impact function $\phi$ and the block liquidation impact function $g$ (see (\ref{payoff7})). The model deals with both the permanent impact and the temporary impact in a regime switching framework. We can apply the dynamic programing principle to derive the HJB equation that involves time variable as well as state variables, which makes the  proofs   more involved than those in \cite{PZY07, PZY08}.
Our main contribution is that we show the optimal value function is the unique continuous viscosity solution to  the HJB equation, which opens the way to solving the problem with the finite difference method.

The paper is organized as follows. Section 2 formulates the optimal liquidation problem and states the main results of the paper. Section 3 gives a numerical  example. Section 4 proves that the optimal value function is continuous (Theorem~\ref{T:Continuity}). Section 5 proves that the value function is the viscosity solution to the HJB equation (Theorem~\ref{T:solution}). Section 6 proves the comparison theorem for the uniqueness of the viscosity solution (Theorem~\ref{T:Comparison}).

\section{Model and Main Results} \label{S:Model}
Let $(\Omega,{\mathcal F},P)$ be a probability space and $({\mathcal F}_{r})_{0\leq r\leq T}$ be the natural filtration generated by a standard Brownian motion process $W$ and
a continuous time Markov chain process $\alpha$, augmented by all $P$-null sets. Assume $W$ and $\alpha$ are independent to each other. Assume that the Markov chain has a
finite state space $\mathbb{M}=\{1,\ldots,m\}$ and is generated by the generator $Q=\{q_{ij}\}$, where $q_{ij}\geq0$ for $i,j\in\mathbb{M}$, $j\neq i$ and
$\sum_{j=1}^m q_{ij}=0$ for each $i\in\mathbb{M}$. The transitional
probability is given by
\begin{equation}\label{E:ModelSetup1}
P\{\alpha(t+\Delta)=j|\alpha(t)=i\}=\begin{cases}
q_{ij}\Delta+o(\Delta)&\text{if $j\neq i$,}\\
1+q_{ii}\Delta+o(\Delta)&\text{if $j=i$}
\end{cases}
\end{equation}
for small time interval $\Delta>0$. The continuous time Markov chain $\alpha(r)_{0\leq r\leq T}$ models
the  economic environment which affects the growth rate and the volatility of the asset price.

Let $r\in[t,T]$ be the time variable, where $T$ is the fixed terminal time and $t\in[0,T)$ is the starting
time. Let $S(r)_{0\leq r\leq T}$ denote the stock price and $X(r)_{0\leq r\leq T}$ the number of shares of stock.
Let $u(r)_{0\leq r\leq T}$ denote the rate of selling the stock, which is a control variable decided by the trader. We call $u=\{u(r)\}_{0\leq r\leq T}$ is admissible if it is progressively measurable and $u(r)\in U$ for a compact set $U\subset[0,\infty)$ for all $t\leq r\leq T$. The stock price $S(r)$ follows a stochastic differential equation with regime switching
\begin{equation}\label{E:ModelSetup4}
dS(r)=\mu(r,S(r),u(r),\alpha(r))dr+\sigma(r,S(r),u(r),\alpha(r))dW(r)
\end{equation}
and the stock holding $X(r)$ follows the dynamics
$$
dX(r)=-u(r)dr.
$$
Since the drift and the diffusion terms of $S$ are affected by the trading strategy $u$ there is the permanent impact of liquidation on the asset price. Such an impact may be negligible  for a small trader (when $u$ is small) but can be significant for a large trader (when $u$ is large). We implicitly assume that the asset price $S(r)$ is positive for all $t\leq r\leq T$. A sufficient condition that guarantees this is that $S$ follows a geometric Brownian motion process with drift and diffusion coefficients depending on time, control and Markov state. We denote by $K$ some generic positive constant which may take different values at different places.

\begin{assumption}\label{ass1}
Functions $f=\mu, \sigma$  satisfy,
 for all $t,s\in[0,T]$, $x,y\in\mathbb{R}$, $\upsilon\in[0,\infty)$ and $\ell\in\mathbb{M}$,
that
\begin{align}\label{E:ModelSetup5}
&|f(t,x,\upsilon,\ell)-f(s,y,\upsilon,\ell)|\leq K(|t-s|+|x-y|)
\;\mbox{ and }\; |f(t,x,\upsilon,\ell)|\leq K(1+|x|).
\end{align}
\end{assumption}

 It can be shown, with
Assumption~\ref{ass1}, that for any admissible control process $u\in{\cal U}$ and any initial values
$(t,s,\ell)\in[0,T)\times(0,\infty)\times\mathbb{M}$, there exists a unique solution, denoted by  $\{S_{t,s,\ell}^u(r), t\leq r\leq T\}$, to equation \eqref{E:ModelSetup4}, and that
the following inequalities hold:
\begin{eqnarray}
E\left[\sup_{r\in[t,T]}{|S_{t,s,\ell}^u(r)|^{p}}\right]&\leq& K(1+s^{p}), \quad
p=1,2 \label{E:ModelSetup6} \\
E\left[\left|S_{t,s,\ell}^u(t_2)-S_{t,s,\ell}^u(t_1)\right|\right]&\leq& K(1+s)|t_2-t_1|^{1/2},\quad t_1, t_2\in [t, T] \label{E:ModelSetup7}\\
E\left[\sup_{r\in[t,T]}\left|S_{t,s_1,\ell}^u(r)-S_{t,s_2,\ell}^u(r)\right|\right]&\leq& K\left|s_1-s_2\right|, \quad s_1,s_2\in (0,\infty). \label{E:StockPriceContinuity}
\end{eqnarray}
The proofs of (\ref{E:ModelSetup6}), (\ref{E:ModelSetup7}) and (\ref{E:StockPriceContinuity}) can be found in Mao and Yuan \cite{MY06} with some minor changes to include control processes, see \cite{MY06}, Theorem3.23, Theorem 3.24 and Lemma 3.3.

Similarly,
$\{X_{t,x}^u(r), t\leq r\leq T\}$ denotes the  stock holding and $\{\alpha_{t,\ell}(r), t\leq r\leq T\}$
the Markov chain process.

Suppose a trader starts from time $t$, endowed with initial values $(X(t), S(t), \alpha(t))=(x,s,\ell)\in(0,\infty)\times(0,\infty)\times{\mathbb M}$.
Define a stopping time
$$\tau_0=\inf\{r\geq t:X_{t,x}^u(r)=0\}\wedge T.$$
 This is the first time that $X_{t,x}^u(r)$ exits from $(0,\infty)$ before or at time $T$.
 Since the model is to study
the liquidation strategy, the trader is only allowed to sell stock without buying back. When the number of shares reaches zero before
time $T$ the liquidation stops. Otherwise, it stops at time $T$.

 The expected discounted total payoff associated with  a strategy $u\in\mathcal U$
is defined by
\begin{equation}\label{payoff7}
J(t,x,s,\ell;u)=E\left[\int_t^{\tau_0}e^{-\beta(r-t)}\phi\left(u(r)\right)S_{t,s,\ell}^u(r)dr+e^{-\beta(\tau_0-t)}g(X(\tau_0))S(\tau_0)\right],
\end{equation}
where $\beta>0$ is a discount rate, $\phi$ a function measuring the temporary liquidation effect, $g$ a function measuring the block liquidation effect, and $E$  the conditional expectation given the information set ${\cal F}_t$ which is equivalent to given $X(t)=x$, $S(t)=s$ and $\alpha(t)=\ell$ since the model is Markov. The first term is the expected discounted accumulated cash value from the stock liquidation and the second term is the expected discounted cash value from the block liquidation at time $T$ for any remaining shares of the stock.

\begin{assumption}\label{ass2}
Functions $f=\phi, g$ are continuous
concave increasing on ${\mathbb R}$ and satisfy $f(0)=0$ and $f'(0)=1$. Furthermore, function $g$ is continuously differentiable and satisfies,
for all $x,y\in {\mathbb R}$, that
$$ |g(x)-g(y)|\leq K|x-y|\; \mbox{ and }\; |g'(x)-g'(y)|\leq K|x-y|.$$
\end{assumption}

Note that  in
 a completely liquid market  $\phi(\upsilon)=\upsilon$ and $g(x)=x$, and that
 $f(0)=0$ and $f'(0)=1$ imply $f(x)$ is approximately equal to $x$ when $x$ is close to 0, which means when the trading rate $u$ is small or the amount of stock $X$ is small then there is essentially no transaction cost and the liquidity impact can be ignored.
The objective of the trader is to maximize the expected discounted revenue from stock liquidation.
The value function is defined by
$$
V(t,x,s,\ell)=\sup_{u\in{\mathcal U}}J(t,x,s,\ell;u).
$$

For $\upsilon\in U$ define operators ${\mathcal L}^\upsilon$ and $\mathcal Q$ of the value function $V$
by
$$
\mathcal{L}^\upsilon V(t,x,s,\ell)=-\upsilon\frac{\partial V}{\partial x}(t,x,s,\ell)+\mu(t,s,\upsilon,\ell)\frac{\partial V}{\partial s}(t,x,s,\ell)
+\frac{1}{2}\sigma^2(t,s,\upsilon,\ell)\frac{\partial^2V}{\partial s^2}(t,x,s,\ell),
$$
and
$$
{\mathcal Q}V(t,x,s,\ell)=\sum_{j\neq\ell}q_{\ell j}\left(V(t,x,s,j)-V(t,x,s,\ell)\right).
$$

The HJB equation for the optimal control problem is, for $(t,x,s,\ell)\in[0,T)\times(0,\infty)\times(0,\infty)\times
{\mathbb M}$,
\begin{equation}\label{E:HJB}
\beta V(t,x,s,\ell)-\frac{\partial V}{\partial t}(t,x,s,\ell)-\sup_{\upsilon\in U}\left\{{\mathcal L}^\upsilon V(t,x,s,\ell)+\phi(\upsilon)s\right\}
-{\mathcal Q}V(t,x,s,\ell)=0,
\end{equation}
with the boundary condition
$$
V(t,0,s,\ell)=0
$$
and the terminal condition
$$
V(T,x,s,\ell)=g(x)s.
$$

It is easy to check that the value function is an increasing function with respect to the asset price and the stock holding. It also
has the following continuity property.
\begin{theorem}\label{T:Continuity}
Assume Assumptions \ref{ass1} and \ref{ass2}. Then
 the value function $V(\cdot,\cdot,\cdot,\ell)$ is continuous on $[0,T]\times[0,\infty)\times(0,\infty)$
for $\ell\in{\mathbb M}$.
\end{theorem}

Since we do not know if the value function $V$ is continuously differentiable and cannot discuss the solution to the HJB equation in the classical sense,  we need to introduce the concept of the viscosity solution to the HJB equation.
\begin{definition}\label{D:Viscosity}
A system of continuous functions $V=\{V(\cdot,\cdot,\cdot,\ell)\}_{\ell\in{\mathbb M}}$ on $[0,T)\times(0,\infty)\times(0,\infty)$ is a viscosity
subsolution (resp.~supersolution) of the HJB equation \eqref{E:HJB} if, for any fixed $\ell\in{\mathbb M}$,
$\varphi\in C^{1,1,2}([0,T)\times(0,\infty)\times(0,\infty))$ and $({\bar t},{\bar x},{\bar s})\in[0,T)\times(0,\infty)
\times(0,\infty)$ such that $V(t,x,s,\ell)-\varphi(t,x,s)$ attains its maximum (resp.~minimum) at $({\bar t},{\bar s},{\bar x})$, we have
\begin{equation}\label{E:Viscosity_1}
\beta\varphi({\bar t},{\bar x},{\bar s})-\frac{\partial\varphi}{\partial t}({\bar t},{\bar x},{\bar s})
-\sup_{\upsilon\in U}\left\{{\mathcal L}^\upsilon\varphi({\bar t},{\bar x},{\bar s})+\phi(\upsilon){\bar s}\right\}-{\mathcal Q}V({\bar t},{\bar x},{\bar s},\ell)\leq0;\quad
(\text{resp.}\;\geq0).
\end{equation}
The system of continuous functions $V$ is a viscosity solution if it is both a viscosity subsolution and a viscosity supersolution.
\end{definition}

We have the following result for the value function.
\begin{theorem} \label{T:solution}
Assume Assumptions \ref{ass1} and \ref{ass2}.
Then the value function $V$ is a viscosity solution to the HJB equation \eqref{E:HJB}.
\end{theorem}


One in general has to use some numerical scheme to find the value function. To ensure the numerical solution to the HJB equation is indeed the value function one has to show that the value function is the unique viscosity solution to the HJB equation, which can be achieved by the following comparison theorem.
\begin{theorem}\label{T:Comparison}
Assume Assumptions \ref{ass1} and \ref{ass2}.
Let $U$  be a viscosity subsolution and $V$ a viscosity supersolution to the HJB equation \eqref{E:HJB} and satisfy the polynomial growth condition and
$U(T,x,s,\ell)\leq V(T,x,s,\ell)$ for all $(x,s,\ell)\in (0,\infty)\times(0,\infty)\times{\mathbb M}$. Then $U\leq V$ on
$[0,T)\times(0,\infty)\times(0,\infty)\times{\mathbb M}$.
\end{theorem}


The proofs of Theorems \ref{T:Continuity}, \ref{T:solution}, and \ref{T:Comparison} are given in Sections 4, 5, and 6, respectively. The proofs are technical and lengthy as one would expect with the viscosity solution method.  The further complication in the proofs over the standard diffusion model is that we need to deal with the Markov chain process $\alpha$ and its relation with the diffusion process $S$.


\section{A Numerical Example} \label{S:Numerical}

In this section we give a numerical example  to find the approximation of the value function and the  optimal selling strategy. The finite difference method is one of the most common approximation schemes for viscosity solutions due to its well-known consistency, stability, convergence analysis, in particular in the presence of the monotonicity property, see \cite{CL84} for numerical solutions of HJB equation and  \cite{PZ06}  for a regime switching optimal stopping problem which results in a system of HJB variational inequalities.  We may apply the numerical scheme of \cite{PZ06} to solve our optimal liquidation problem.   The numerical example is   to provide  a snapshot of the  optimal trading strategy at a given specific time.

Assume that there are only two regimes. Regime 1 represents the strong economy and regime 2 the weak economy and assume that the stock price $S(r)$
follows a GBM process with
$\mu(r,s,u,\alpha)=\mu(\alpha)s$ and
$\sigma(r,s,u,\alpha)=\sigma(\alpha)s$.
Define  variables $z=\log{s}$ and $\tau=T-t$ and a function
$W(\tau,x,z,\ell)=V(t,x,s,\ell)$. The HJB equation~\eqref{E:HJB} becomes
\begin{align}\label{E:HJB2}
&\beta W(\tau,x,z,\ell)+\frac{\partial W}{\partial\tau}(\tau,x,z,\ell)-\sup_{\upsilon\in U}\Big\{
-\upsilon\frac{\partial W}{\partial x}(\tau,x,z,\ell)+\mu(\ell)\frac{\partial W}{\partial z}(\tau,x,z,\ell)\notag\\
&+\frac{1}{2}\sigma^2(\ell)\left(\frac{\partial^2W}{\partial z^2}(\tau,x,z,\ell)-\frac{\partial W}{\partial z}(\tau,x,z,\ell)\right)
+\phi(\upsilon)e^z\Big\}-{\mathcal Q}W(\tau,x,z,\ell)=0,
\end{align}
with the boundary condition $W(\tau,0,z,\ell)=0$ and the terminal condition
$W(0,x,z,\ell)=g(x)e^z$.

To approximate the solution to \eqref{E:HJB2} we discretize variables $\tau$, $x$ and $z$ with stepsizes $\Delta \tau, \Delta x, \Delta z$, respectively.   The value of  $W$
at a grid point $(\tau_n, x_i, z_j)$ in the regime $\ell$ is denoted by $W^n_{i,j}(\ell)$. The derivatives of $W$ are approximated by
$W_\tau=(W^{n+1}_{i,j}(\ell)-W^{n}_{i,j}(\ell))/\Delta \tau$,
$W_x=(W^{n}_{i+1,j}(\ell)-W^{n}_{i-1,j}(\ell))/(2\Delta x)$,
$W_z=(W^{n}_{i,j+1}(\ell)-W^{n}_{i,j-1}(\ell))/(2\Delta z)$, and
$W_{zz}=(W^{n}_{i,j+1}(\ell)+W^{n}_{i,j-1}(\ell)-2W^{n}_{i,j}(\ell))/\Delta z^2$.
Discretizing equation (\ref{E:HJB2}) and rearranging the terms,
we have
\begin{align}\label{E:HJB3}
W^{n+1}_{i,j}(\ell)=&\Delta \tau\Bigg[\left(-\beta+\frac{1}{\Delta \tau}-q_{\ell\ell'}-\frac{\sigma(\ell)^2}{\Delta z^2}\right)W^{n}_{i,j}(\ell)
+\left(\frac{\mu(\ell)-\frac{1}{2}\sigma(\ell)^2}{2\Delta z}+\frac{\sigma(\ell)^2}{2\Delta z^2}\right)W^{n}_{i,j+1}(\ell)\notag\\
&{}+\left(-\frac{\mu(\ell)-\frac{1}{2}\sigma(\ell)^2}{2\Delta z}+\frac{\sigma(\ell)^2}{2\Delta z^2}\right)W^{n}_{i,j-1}(\ell)+q_{\ell\ell'}W^{n}_{i,j}(\ell')\notag\\
&{}+\sup_{\upsilon\in U}\Big\{-\upsilon\frac{W^{n}_{i+1,j}(\ell)-W^{n}_{i-1,j}(\ell)}{2\Delta x}+\phi(\upsilon)e^z\Big\}\Bigg],
\end{align}
where    $\ell, \ell'=1,2$ and $\ell\ne \ell'$. Assume that the temporary liquidation impact function is given by
$$
\phi(\upsilon)={1\over \alpha}(1-e^{-\alpha\upsilon}),
$$
where $\alpha>0$, and  the block liquidation impact function is given by
\begin{equation}\label{E:piecewise2}
g(x)=\begin{cases}
x,&\text{if $0\leq x\leq5$,}\\
-0.01x^2+1.1x-0.25,&\text{if $5<x\leq15$,}\\
10+0.8(x-10),&\text{if $15<x\leq40$,}\\
-0.0075x^2+1.4x-10,&\text{if $40<x\leq60$,}\\
42+0.5(x-50),&\text{if $x>60$.}
\end{cases}\notag
\end{equation}
Functions $\phi$ and $g$ satisfy Assumption~\ref{ass2}. In fact, $g$ is constructed as a smooth approximation to a function $f$ defined by
\begin{equation}
f(x)=\begin{cases}
x,&\text{if $0\leq x\leq 10$,}\\
10+0.8(x-10),&\text{if $10<x\leq 50$,}\\
42+0.5(x-50),&\text{if $x>50$.}
\end{cases}\notag
\end{equation}
Function $f$ captures the block liquidation effect  at time $T$ but is not differentiable   at $x=10$ and 50 and does not satisfy Assumption~\ref{ass2}.

Data used for numerical tests are $\alpha=0.005$, $\beta=0.01$, $\{\mu(1),\mu(2)\}=\{0.3,-0.1\}$, $\{\sigma(1),\sigma(2)\}=\{0.2,0.4\}$, $q_{12}=0.5$, $q_{21}=1$, $\upsilon\in U=[0,100]$,
$t\in[0,1]$, $x\in[0,100]$, $s\in[e^{-1},e^2]$.

\begin{figure}[htb]
  \centering
  \subfigure[$\ell=1$]{\includegraphics[width=7cm]{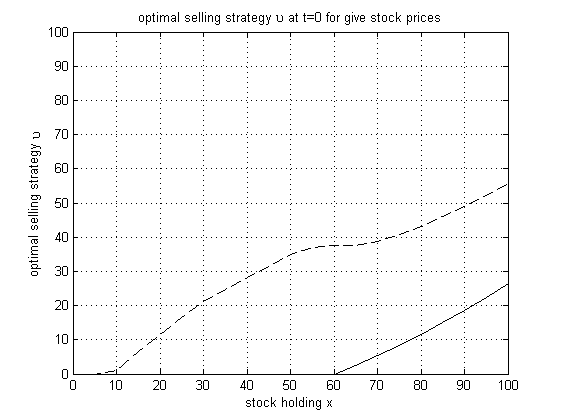}}
  \caption{The optimal control at time $t=0$ against stock holding $x$. The solid line is for regime 1  and the dashed line for regime 2.}\label{fig:2DContr}
\end{figure}

Figure~\ref{fig:2DContr} demonstrates the relationship between the optimal selling strategy and the stock holding. It is clear that the more shares
one holds, the sooner and the more one wants to sell to avoid the potential large transaction cost during the whole period. The market regime determines at what level of stock holding one should start to sell. In a rising market (regime 1) the trader is willing to keep the stock
for a longer period in the hope for a higher price, which results in a lower optimal selling rate, whereas in a falling market (regime 2) the trader wants to liquidate the stock quickly to avoid a lower price. This is consistent with the general market phenomenon. The optimal trading strategy is independent of initial asset price in the numerical test, which is not surprising as the asset price follows a GBM process and depends on the initial asset price linearly. In general, the optimal trading strategy should also depend on the asset price.
The particular shape of the curve
in Figure~\ref{fig:2DContr} is determined by the tradeoff between  function $\phi$ that captures the liquidity effect from 'flow' trading and function
$g$ that reflects the  transaction cost for the block liquidation at the terminal time.
Note that if there is no temporary pricing impact on liquidation, i.e., $\phi(\upsilon)=\upsilon$, then the optimal liquidation strategy is a ``bang-bang'' control with either no trading $\upsilon=0$ or selling at maximum rate $\upsilon=100$ due to the linear dependence of control $\upsilon$ in the Hamiltonian function.


\section{Proof of Theorem~\ref{T:Continuity}} \label{S:Continuity}
We first convert the original control problem into a problem without terminal bequest function. Since function $g$ is continuously
differentiable, we can apply Dynkin's formula to
$e^{\beta(\tau_0-t)}g(X^u_{t,x}(\tau_0))S^u_{t,s,\ell}(\tau_0)$ and rewrite the total payoff $J$ as
$$J(t,x,s,\ell; u)=
g(x)s+E\left[\int_t^{\tau_0}
L(r,X^u_{t,x}(r),S^u_{t,s,\ell}(r),u(r),\alpha_{t,\ell}(r))dr\right],
$$
where
$$
L(r,x,s,\upsilon,\alpha)=e^{-\beta(r-t)}\left[\phi(\upsilon)s-\beta g(x)s-\upsilon g'(x)s+\mu(r,s,\upsilon,\alpha)g(x)\right].
$$

Define a new value function by
$$
\widetilde{V}(t,x,s,\ell)=\sup_{u\in{\mathcal U}}E\left[\int_{t}^{\tau_0}L\left(r,X^u_{t,x}(r),S^u_{t,s,\ell}(r),u(r),\alpha_{t,\ell}(r)\right)dr\right].
$$
Since $V(t,x,s,\ell)=\widetilde{V}(t,x,s,\ell)+g(x)s$, we know
$V(t,x,s,\ell)$ is continuous as long as $\widetilde{V}(t,x,s,\ell)$ is continuous. From now on in this section we work on the value function $\widetilde{V}$.

To prove the continuity of $\tilde{V}$ we need to define some perturbed problems and show their corresponding value functions are continuous and  converge quasi-uniformly to $\widetilde{V}$, which establishes Theorem~\ref{T:Continuity}.

For $0<\epsilon<1$ define the stopping time
$$\tau_\epsilon=\inf\{r\geq t: X_{t,x}^u(r)=-\epsilon\}\wedge T,$$
 which is the
first time $X_{t,x}^u(r)$ exits from $(-\epsilon,\infty)$.
A control process $u=\{u(r)\}_{0\leq r\leq T}$ is admissible if it is progressively measurable and $u(r)\in U_\epsilon(X_{t,x}^u(r))$, where $U_\epsilon(x)=U$ if $x\geq 0$ and
$U_\epsilon(x)=\hat U$, a compact subset of $U$ in $(0,\infty)$, if $x<0$.
The key here is to rule out zero from the
compact set $\hat{U}$ after $X(r)$ reaches zero. The admissible control set is the collection of all admissible controls, denoted by   ${\mathcal U}_\epsilon$. Note that when we only look at the control process before $\tau_0$, the two admissible control sets, ${\mathcal U}$ and ${\mathcal U}_\epsilon$, are the same.

To simplify the notation denote by
$$ L^u_{t,x,s,\ell}(r):=L\left(r,X^u_{t,x}(r),S^u_{t,s,\ell}(r),u(r),\alpha_{t,\ell}(r)\right).$$
Since  $U$ is a compact set in $[0,\infty)$, say $[0,N]$, we know that  $X^u_{t,x}(r)\in [x-NT, x]$ for $t\leq r\leq T$, which implies that $|g(X^u_{t,x}(r))|$ and   $|g'(X^u_{t,x}(r))|$ are bounded by some constant $K_x$ depending on $x$ due to continuity of $g$ and $g'$. Assumptions \ref{ass1} and \ref{ass2}
 imply that, for $t\leq r\leq T$,
\begin{align}\label{E:L-inequality}
|L^u_{t,x,s,\ell}(r)|\leq K_x\left(1+S^u_{t,s,\ell}(r)\right)
\end{align}
and
\begin{equation} \label{ddag}
\left|L^u_{t,x_1,s_1,\ell}(r)-
L^u_{t,x_2,s_2,\ell}(r)\right|
\leq K_{x_1}|S^u_{t,s_1,\ell}(r)-S^u_{t,s_2,\ell}(r)|+K\left(1+S^u_{t,s_2,\ell}(r)\right)|x_1-x_2|
\end{equation}
for some constant $K_{x_1}$ depending on $x_1$.

\begin{remark} \label{rk1}
{\rm
In the proof we need to estimate $|L^u_{t,x,s,\ell}(r)|$ several times for different $x$. One case is that $x=-\epsilon$ for $0<\epsilon<1$. Then $X^u_{t,x}(r)\in [-1-NT, 0]$ and constant $K_x$ can be replaced by a generic constant $K$ independent of $x$. The other case is that $x$ is within a distance $d$ of another point $x_1$. Then $X^u_{t,x}(r)\in [x_1-d-NT, x_1+d]$ and constant $K_x$ can be written as $K_{x_1}$ depending on $x_1$ for all such $x$.
}\end{remark}

For $\epsilon\in (0,1)$ define a perturbed value function by
$$
\widetilde{V}^{\epsilon}(t,x,s,\ell)=\sup_{u\in{\mathcal U}_\epsilon}E\left[\int_t^{\tau_\epsilon}L^u_{t,x,s,\ell}(r)dr\right].
$$

For $\rho>0$ define an auxiliary function
$$
\Gamma^{\epsilon,\rho,u}_{t,x}(r)=\exp\left(-\frac{1}{\rho}\left(X_{t,x}^u(r)+\epsilon\right)^-\right),
$$
where $x^-=\max(0, -x)$. Clearly, we have
 $\Gamma^{\epsilon,\rho,u}_{t,x}(r)\leq 1$ and,
by the definition of the stopping time $\tau_{\epsilon}$,
 $\Gamma^{\epsilon,\rho,u}_{t,x}(r)=1$ for $r\in[t,\tau_{\epsilon}]$.
The auxiliary value function $\widetilde{V}^{\epsilon,\rho}$ is defined by
\begin{align}
\widetilde{V}^{\epsilon,\rho}(t,x,s,\ell)&=\sup_{u\in{\mathcal U}_\epsilon}\widetilde{J}^{\epsilon,\rho}(t,x,s,\ell;u)
:=E\left[\int_t^T\Gamma^{\epsilon,\rho,u}_{t,x}(r)L^u_{t,x,s,\ell}(r)dr\right].\notag
\end{align}
From (\ref{E:L-inequality}) and (\ref{E:ModelSetup6})  we have that
\begin{equation}
|\widetilde{V}^{\epsilon,\rho}(t,x,s,\ell)|\leq
K_x(1+s). \label{dddag}
\end{equation}

\begin{lemma}\label{Lemma2}
 $\widetilde{V}^{\epsilon,\rho}(t,x,s,\ell)$ converges to  $\widetilde{V}(t,x,s,\ell)$ quasi-uniformly as $\rho\to0$ and $\epsilon\to0$.
\end{lemma}

\begin{proof}
\textit{Step 1.} Fix a point $(t,x,s)\in[0,T)\times\{-\epsilon\}\times(0,\infty)$.
Since $X^u_{t,x}(t)=-\epsilon$ we have $\tau_{\epsilon}=t$ and for $r> t$
the admissible control $u(r)$ is in a compact set $U_{\epsilon}(x):=[N_0,N]\subset U$ with $N_0>0$, which implies that $X^u_{t,-\epsilon}(r)<-\epsilon$ and
\begin{equation} \label{dag}
\exp\left(-\frac{N}{\rho}(r-t)\right)\leq \Gamma^{\epsilon,\rho,u}_{t,-\epsilon}(r)=\exp\left(-\frac{1}{\rho}\int_{t}^ru(s)ds\right)
\leq \exp\left(-\frac{N_0}{\rho}(r-t)\right)
\end{equation}
and
$\lim_{\rho\to0}\Gamma^{\epsilon,\rho,u}_{t,-\epsilon}(r)=0$.
(\ref{dag}), (\ref{E:L-inequality}) and (\ref{E:ModelSetup6}) imply that, also noting Remark \ref{rk1},
\begin{align}
\widetilde{J}^{\epsilon,\rho}\left(t,-\epsilon,s,\ell;u\right)
&\leq K\int_t^Te^{-{N_0\over \rho}(r-t)}\left(1+E\left[S_{t,s,\ell}^u(r)\right]\right)dr\leq K(1+s){\rho\over N_0}\left(1-e^{-{N_0\over \rho}T}\right).\notag
\end{align}
Similarly, we have
$$
\widetilde{J}^{\epsilon,\rho}\left(t,-\epsilon,s,\ell;u\right)
\geq -K(1+s){\rho\over N}\left(1-e^{-{N\over \rho}T}\right).\notag
$$
Combining the above two inequalities and taking the supremum, we have
$$
-K(1+s){\rho\over N}\left(1-e^{-{N\over \rho}T}\right)\leq
\widetilde{V}^{\epsilon,\rho}\left(t,-\epsilon,s,\ell\right)
\leq K(1+s){\rho\over N_0}\left(1-e^{-{N_0\over \rho}T}\right).
$$

Applying the dynamic programming principle (see \cite{FS06}), for $(t,x,s,\ell)\in[0,T)\times[0,\infty)\times(0,\infty)\times{\mathbb M}$, we have
\begin{align}
\widetilde{V}^{\epsilon,\rho}(t,x,s,\ell)
&=\sup_{u\in{\mathcal U}_\epsilon}E\left[\int_t^{\tau_\epsilon}L^u_{t,x,s,\ell}(r)dr+
e^{-\beta(\tau_\epsilon-t)}\widetilde{V}^{\epsilon,\rho}\left(\tau_\epsilon,-\epsilon,S_{t,s,\ell}^u(\tau_\epsilon),\alpha_{t,\ell}(\tau_\epsilon)\right)\right]\notag\\
&\leq\sup_{u\in{\mathcal U}_\epsilon}E\left[\int_t^{\tau_\epsilon}L^u_{t,x,s,\ell}(r)dr+
K(1+S_{t,s,\ell}^u(\tau_\epsilon)){\rho\over N_0}\left(1-e^{-{N_0\over \rho}T}\right)\right]\notag\\
&\leq\widetilde{V}^\epsilon(t,x,s,\ell)+K(1+s){\rho\over N_0}\left(1-e^{-{N_0\over \rho}T}\right).\notag
\end{align}
Similarly, we have
$$
\widetilde{V}^{\epsilon,\rho}(t,x,s,\ell)
\geq\widetilde{V}^\epsilon(t,x,s,\ell)-K(1+s){\rho\over N}\left(1-e^{-{N\over \rho}T}\right).\notag
$$
The above two inequalities imply that
$\widetilde{V}^{\epsilon,\rho}(t,x,s,\ell)$ converges to $\widetilde{V}^{\epsilon}(t,x,s,\ell)$ quasi-uniformly as $\rho\to0$, independent of $\epsilon$.

\textit{Step 2.} By the definition of the perturbed value function, the Cauchy-Schwartz inequality, (\ref{E:L-inequality}) and (\ref{E:ModelSetup6}), we  have
\begin{align}
\widetilde{V}^\epsilon(t,x,s,\ell)
=&\sup_{u\in{\mathcal U}_\epsilon}E\left[\int_t^{\tau_0}L^u_{t,x,s,\ell}(r)dr+
\int_{\tau_0}^{\tau_\epsilon}L^u_{t,x,s,\ell}(r)dr\right]\notag\\
\leq&\widetilde{V}(t,x,s,\ell)+\sup_{u\in{\mathcal U}_\epsilon}E\left[\int_t^T{\mathbbm 1}_{\{\tau_0<r<\tau_\epsilon\}}
L^u_{t,x,s,\ell}(r)dr\right]\notag\\
\leq&\widetilde{V}(t,x,s,\ell)+\sup_{u\in{\mathcal U}_\epsilon}\sqrt{E[\tau_\epsilon-\tau_0]}\sqrt{E\left[\int_t^T
L^u_{t,x,s,\ell}(r)^2dr\right]}\notag\\
\leq&\widetilde{V}(t,x,s,\ell)+K_{x}(1+s)\left(\frac{\epsilon}{N_0}\right)^{1/2}.\notag
\end{align}
for some constant $K_{x}$ depending on $x$. Similarly, we have
$$
\widetilde{V}^\epsilon(t,x,s,\ell)\geq\widetilde{V}(t,x,s,\ell)-K_{x}(1+s)\left(\frac{\epsilon}{N_0}\right)^{1/2}.
$$

As $\epsilon\to0$, $\widetilde{V}^\epsilon(t,x,s,\ell)$ converges to $\widetilde{V}(t,x,s,\ell)$ quasi-uniformly. Combining the results of
Steps 1 and  2, we conclude that $\widetilde{V}^{\epsilon,\rho}(t,x,s,\ell)$ converges to $\widetilde{V}(t,x,s,\ell)$ quasi-uniformly
as $\rho\to0$ and $\epsilon\to0$.
\end{proof}

\begin{lemma}\label{Lemma2a}
$\widetilde{V}^{\epsilon,\rho}(\cdot,\cdot,\cdot,\ell)$ is continuous on $[0,T]\times[0,\infty)\times(0,\infty)$
for $\ell\in{\mathbb M}$ and arbitrary constants $\epsilon>0$ and $\rho>0$.
\end{lemma}
\begin{proof}
\textit{Step 1.} Let $(x_1,s_1), (x_2,s_2)\in [0,\infty)\times(0,\infty)$ satisfying $|x_2-x_1|\leq 1$ and $|s_2-s_1|\leq 1$ and  $t\in[0,T]$ and $\ell\in\mathbb{M}$. Consider the auxiliary value functions
$\widetilde{V}^{\epsilon,\rho}(t,x_1,s_1,\ell)$ and $\widetilde{V}^{\epsilon,\rho}(t,x_2,s_2,\ell)$.

Since $|e^{-a}-e^{-b}|\leq|a-b|$ for any $a,b\geq 0$, we have
\begin{equation} \label{E:Theorem1_3.2}
\left|\Gamma^{\epsilon,\rho,u}_{t,x_1}(r)-\Gamma^{\epsilon,\rho,u}_{t,x_2}(r)\right|\leq\frac{1}{\rho}\left| (X_{t,x_1}^u(r)+\epsilon)^--(X_{t,x_2}^u(r)+\epsilon)^-\right|\leq \frac{1}{\rho}\left|x_1-x_2\right|.
\end{equation}

By the definition of $\widetilde{V}^{\epsilon,\rho}$ and the relation $|\sup A - \sup B| \leq \sup |A-B|$ we have
\begin{align}
&\left|\widetilde{V}^{\epsilon,\rho}(t,x_1,s_1,\ell)-\widetilde{V}^{\epsilon,\rho}(t,x_2,s_2,\ell)\right|\notag\\
\leq&\sup_{u\in{\mathcal U}_\epsilon}E\Bigg[\int_t^T\left|\Gamma^{\epsilon,\rho,u}_{t,x_1}(r)L^u_{t,x_1,s_1,\ell}(r)-\Gamma^{\epsilon,\rho,u}_{t,x_2}(r)L^u_{t,x_2,s_2,\ell}(r)\right|dr\Bigg]\notag\\
\leq&\sup_{u\in{\mathcal U}_\epsilon}E\Bigg[\int_t^T(\left|L^u_{t,x_1,s_1,\ell}(r)
-L^u_{t,x_2,s_2,\ell}(r)\right|
+ \left|L^u_{t,x_2,s_2,\ell}(r)\left(\Gamma^{\epsilon,\rho,u}_{t,x_1}(r)-\Gamma^{\epsilon,\rho,u}_{t,x_2}(r)\right)\right|) dr\Bigg]\notag\\
\leq&K_{x_1}|s_1-s_2|+K\left(1+s_2\right)|x_1-x_2|+\frac{1}{\rho}|x_1-x_2|K_{x_1}\left(1+s_2\right)\notag\\
\leq&K_{x_1, s_1}(|x_1-x_2|+|s_1-s_2|),\label{star}
\end{align}
where $K_{x_1,s_1}$ is some constant depending on $x_1$ and $s_1$. In the second last inequality we have used (\ref{ddag}),  (\ref{E:StockPriceContinuity}), (\ref{E:ModelSetup6}),
(\ref{E:L-inequality}), (\ref{E:Theorem1_3.2}) and Remark \ref{rk1}.
This shows that
the auxiliary value function $\widetilde{V}^{\epsilon,\rho}(t,x,s,\ell)$ is continuous in $(x,s)$, uniformly in $t$.

\textit{Step 2.} We  prove that the auxiliary value function $\widetilde{V}^{\epsilon,\rho}$ is continuous in $t$. Let
$0\leq t_1<t_2\leq T$ and $(x,s,\ell)\in[0,\infty)\times(0,\infty)\times{\mathbb M}$. By the dynamic programming principle, for any $\delta>0$, there exists an admissible  control $u_\delta\in{\mathcal U}_\epsilon$ such that
\begin{align}
&\widetilde{V}^{\epsilon,\rho}(t_1,x,s,\ell)-\delta\notag\\
\leq&E\left[\int_{t_1}^{t_2}\Gamma^{\epsilon,\rho,u_\delta}_{t_1,x}(r)L^{u_\delta}_{t_1,x,s,\ell}(r)dr
+e^{-\beta(t_2-t_1)}\widetilde{V}^{\epsilon,\rho}(t_2,X_{t_1,x}^{u_\delta}(t_2),S_{t_1,s,\ell}^{u_\delta}(t_2),\alpha_{t_1,\ell}(t_2))\right]\notag\\
\leq&\widetilde{V}^{\epsilon,\rho}(t_1,x,s,\ell).\notag
\end{align}
Rearranging the above inequalities,  we have
\begin{align}
&\left|\widetilde{V}^{\epsilon,\rho}(t_1,x,s,\ell)-\widetilde{V}^{\epsilon,\rho}(t_2,x,s,\ell)\right|-\delta\notag\\
\leq&\left|E\left[\int_{t_1}^{t_2}\Gamma^{\epsilon,\rho,u_\delta}_{t_1,x}(r)L^{u_\delta}_{t_1,x,s,\ell}(r)dr
+e^{-\beta(t_2-t_1)}
\widetilde{V}^{\epsilon,\rho}(t_2,X_{t_1,x}^{u_\delta}(t_2),S_{t_1,s,\ell}^{u_\delta}(t_2),\alpha_{t_1,\ell}(t_2))\right]\right.-\widetilde{V}^{\epsilon,\rho}(t_2,x,s,\ell)\bigg|\notag\\
\leq&E\left[\int_{t_1}^{t_2}\left|L^{u_\delta}_{t_1,x,s,\ell}(r)\right|dr\right]+E\left[\left|e^{-\beta(t_2-t_1)}\widetilde{V}^{\epsilon,\rho}(t_2,X_{t_1,x}^{u_\delta}(t_2),S_{t_1,s,\ell}^{u_\delta}(t_2),\ell)
-\widetilde{V}^{\epsilon,\rho}(t_2,x,s,\ell)\right|\right]\notag\\
&+E\left[\left|\widetilde{V}^{\epsilon,\rho}(t_2,X_{t_1,x}^{u_\delta}(t_2),S_{t_1,s,\ell}^{u_\delta}(t_2),\alpha_{t_1,\ell}(t_2))
-\widetilde{V}^{\epsilon,\rho}(t_2,X_{t_1,x}^{u_\delta}(t_2),S_{t_1,s,\ell}^{u_\delta}(t_2),\ell)\right|\right]\notag\\
=&I_1+I_2+I_3.\notag
\end{align}

 (\ref{E:L-inequality}) and (\ref{E:ModelSetup6}) imply that
\begin{align}
I_1\leq&K_x(1+s) (t_2-t_1).\notag
\end{align}
(\ref{dddag}) and Remark \ref{rk1} imply that
$$
E\left[\widetilde{V}^{\epsilon,\rho}\left(t_2,X_{t_1,x}^{u_\delta}(t_2),S_{t_1,s,\ell}^{u_\delta}(t_2),\alpha_{t_1,\ell}(t_2)\right)\right]
\leq E\left[K_x\left(1+S_{t_1,s,\ell}^{u_\delta}(t_2)\right)\right]\leq K_{x,s}
$$
for some constant $K_{x,s}$ depending on $x$ and $s$. Noting that the term inside the expectation of $I_3$ is zero when $\alpha_{t_1,\ell}(t_2)=\ell$, using Cauchy-Schwartz inequality and combining the above inequality, we have
\begin{align}
I_3
\leq&K_{x,s}\sqrt{P\left[\alpha_{t_1,\ell}(t_2)\neq\ell\right]}.\notag
\end{align}
Using (\ref{star})  and \eqref{E:ModelSetup7}, we have
\begin{align}
I_2
\leq&E\left[\left|\widetilde{V}^{\epsilon,\rho}\left(t_2,X_{t_1,x}^{u_\delta}(t_2),S_{t_1,s,\ell}^{u_\delta}(t_2),\ell\right)
-\widetilde{V}^{\epsilon,\rho}\left(t_2,x,s,\ell\right)\right|+\left|(e^{-\beta(t_2-t_1)}-1)\widetilde{V}^{\epsilon,\rho}(t_2,x,s,\ell)\right|\right]\notag\\
\leq&K_{x,s}(E\left[\left|X_{t_1,x}^{u_\delta}(t_2)-x\right|\right]+E\left[\left|S_{t_1,s,\ell}^{u_\delta}(t_2)-s\right|\right])+
E\left[\widetilde{V}^{\epsilon,\rho}(t_2,x,s,\ell)\right]\left|e^{-\beta(t_2-t_1)}-1\right|\notag\\
\leq&K_{x,s}\left((t_2-t_1)+(t_2-t_1)^{1/2}+\left|e^{-\beta(t_2-t_1)}-1\right|\right)\notag
\end{align}
for some constant $K_{x,s}$ depending on $x,s$.
The above estimates for $I_1, I_2,I_3$ show that they all tend to 0 as $t_2- t_1$ tends to 0, independent of  $\delta$ and control $u_\delta$ but dependent on $x$ and $s$. Therefore,
$$\left|\widetilde{V}^{\epsilon,\rho}(t_1,x,s,\ell)-\widetilde{V}^{\epsilon,\rho}(t_2,x,s,\ell)\right|-\delta \to 0
\mbox{ as } t_2-t_1\to 0.$$
The arbitrariness of $\delta$ confirms that $\widetilde{V}^{\epsilon,\rho}(t,x,s,\ell)$  is
 continuous in $t$.

Combining the results of Steps 1 and  2, we conclude that $\widetilde{V}^{\epsilon,\rho}(\cdot,\cdot,\cdot,\ell)$ is continuous
in $(t,x,s)$ for each $\ell\in{\mathbb M}$.
\end{proof}

By Lemmas \ref{Lemma2} and \ref{Lemma2a}, the auxiliary value function $\widetilde{V}^{\epsilon,\rho}(t,x,s,\ell)$ converges quasi-uniformly
to the value function $\widetilde{V}(t,x,s,\ell)$ as $\epsilon\to0$ and $\rho\to0$ and $\widetilde{V}^{\epsilon,\rho}(t,x,s,\ell)$ is continuous in $(t,x,s)$,  which shows that
$\widetilde{V}(t,x,s,\ell)$ is continuous on $[0,T]\times[0,\infty)\times(0,\infty)$ for each $\ell\in{\mathbb M}$. We have proved Theorem~\ref{T:Continuity}.

\section{Proof of Theorem \ref{T:solution}}
We first show that $V$ is a viscosity supersolution.

\begin{theorem} \label{T:supersolution}
Given Assumption~\ref{ass1}, the value function $V=\left\{V(t,x,s,\ell)\right\}_{\ell\in{\mathbb M}}$ is a viscosity supersolution
of the HJB equation \eqref{E:HJB}.
\end{theorem}

\begin{proof}
Let $\ell\in{\mathbb M}$, $(\bar{t},\bar{x},\bar{s})\in[0,T)\times(0,\infty)\times(0,\infty)$. Let the test function
$\varphi(t,x,s)\in C^{1,1,2}([0,T)\times(0,\infty)\times(0,\infty))$ such that $V(t,x,s,\ell)-\varphi(t,x,s)$ attains
its minimum at $(\bar{t},\bar{x},\bar{s})$ and, without loss of generality, $V(\bar{t},\bar{x},\bar{s},\ell)-\varphi
(\bar{t},\bar{x},\bar{s})=0$. Choose a constant control $\bar{u}(t)\equiv\upsilon\in U$ for $t\in[0,\tau_0]$. Let the state variables $X$ and $S$
start from time $\bar{t}$ with initial values $\bar{x}$ and $\bar{s}$.

Define $\hat{\tau}_1$ as the first jump time of the regime $\alpha_{\bar{t},\ell}(\cdot)$. Without loss of generality, assume that $\eta$ is small enough such that
$B_\eta(\bar{x},\bar{s})\subset(0,\infty)\times(0,\infty)$. Define $\hat{\tau}_2$ by
$$
\hat{\tau}_2:=\inf\left\{r\geq\bar{t}:\left(X_{\bar{t},\bar{x}}^{\bar u}(r),S_{\bar{t},\bar{s},\ell}^{\bar u}(r)\right)\not\in B_\eta\left(\bar{x},\bar{s}
\right)\right\}.
$$

For $h<T-\bar{t}$, define the stopping time $\tau:=(\bar{t}+h)\wedge\hat{\tau}_1\wedge\hat{\tau}_2$. Note that $\tau<\tau_0$. By dynamic programming principle,
\begin{equation}\label{E:Supersolution_2}
V(\bar{t},\bar{x},\bar{s},\ell)\geq E\left[\int_{\bar{t}}^{\tau}e^{-\beta(r-\bar{t})}\phi\left(\bar{u}(r)\right)S_{\bar{t},\bar{s},\ell}^{\bar u}(r)dr
+e^{-\beta(\tau-\bar{t})}V\left(\tau,X_{\bar{t},\bar{x}}^{\bar u}(\tau),S_{\bar{t},\bar{s},\ell}^{\bar u}(\tau),\alpha_{\bar{t},\ell}(\tau)\right)\right].
\end{equation}

Define
\begin{align}\label{E:Supersolution_3}
\psi(t,x,s,i)=\begin{cases}
\varphi(t,x,s)&\text{if $i=\ell$},\\
V(t,x,s,i) &\text{if $i\neq\ell$}.
\end{cases}
\end{align}
Applying Dynkin's formula at point $(\bar{t},\bar{x},\bar{s},\ell)$, also noting $\psi(t,x,s,\ell)=\varphi(t,x,s)$,
  we have
\begin{align}
&E\left[e^{-\beta(\tau-\bar{t})}\psi\left(\tau,X_{\bar{t},\bar{x}}^{\bar u}(\tau),S_{\bar{t},\bar{s},\ell}^{\bar u}(\tau)
,\alpha_{\bar{t},\ell}(\tau)\right)\right]\notag\\
=&\varphi\left(\bar{t},\bar{x},\bar{s}\right)+
E\left[\int_{\bar{t}}^{\tau}\left\{(-\beta)e^{-\beta(r-\bar{t})}\varphi
\left(r,X_{\bar{t},\bar{x}}^{\bar u}(r),S_{\bar{t},\bar{s},\ell}^{\bar u}(r)\right)\right.\right.\notag\\
&\left.\left.+e^{-\beta(r-\bar{t})}\left(\frac{\partial}{\partial t}+\mathcal{L}^{\upsilon}\right)\varphi
\left(r,X_{\bar{t},\bar{x}}^{\bar u}(r),S_{\bar{t},\bar{s},\ell}^u(r)\right)+{\mathcal Q}\psi
\left(r,X_{\bar{t},\bar{x}}^{\bar u}(r),S_{\bar{t},\bar{s},\ell}^{\bar u}(r),\ell\right)\right\}dr\right],\label{E:Supersolution_8}
\end{align}
which implies, from the choice of $(\bar{t},\bar{x},\bar{s},\ell)$ and the definition of $\psi$, that
\begin{align}
&E\left[e^{-\beta(\tau-\bar{t})}\psi\left(\tau,X_{\bar{t},\bar{x}}^{\bar u}(\tau),S_{\bar{t},\bar{s},\ell}^{\bar u}(\tau)
,\alpha_{\bar{t},\ell}(\tau)\right)\right]\notag\\
\geq&
V\left(\bar{t},\bar{x},\bar{s},\ell\right)+
E\left[\int_{\bar{t}}^{\tau}\left\{(-\beta)e^{-\beta(r-\bar{t})}\varphi
\left(r,X_{\bar{t},\bar{x}}^{\bar u}(r),S_{\bar{t},\bar{s},\ell}^{\bar u}(r)\right)\right.\right.\notag\\
&\left.\left.+e^{-\beta(r-\bar{t})}\left(\frac{\partial}{\partial t}+\mathcal{L}^{\upsilon}\right)\varphi
\left(r,X_{\bar{t},\bar{x}}^{\bar u}(r),S_{\bar{t},\bar{s},\ell}^u(r)\right)+{\mathcal Q}V
\left(r,X_{\bar{t},\bar{x}}^{\bar u}(r),S_{\bar{t},\bar{s},\ell}^{\bar u}(r),\ell\right)\right\}dr\right].\label{E:Supersolution_9}
\end{align}
Substitute \eqref{E:Supersolution_9} into \eqref{E:Supersolution_2} and divide both sides by $-h$ we get
\begin{align}\label{E:Supersolution_10}
0\leq&E\left[\frac{1}{h}\int_{\bar{t}}^{\tau}\left\{e^{-\beta(r-\bar{t})}\left(\beta\varphi
\left(r,X_{\bar{t},\bar{x}}^{\bar u}(r),S_{\bar{t},\bar{s},\ell}^{\bar u}(r)\right)
-\left(\frac{\partial}{\partial t}+\mathcal{L}^{\upsilon}\right)\varphi
\left(r,X_{\bar{t},\bar{x}}^{\bar u}(r),S_{\bar{t},\bar{s},\ell}^{\bar u}(r)\right)\right.\right.\right.\notag\\
&\left.\left.-\phi\left({\bar u}(r)\right)S_{\bar{t},\bar{s},\ell}^{\bar u}(r)\bigg)-{\mathcal Q}V
\left(r,X_{\bar{t},\bar{x}}^{\bar u}(r),S_{\bar{t},\bar{s},\ell}^{\bar u}(r),\ell\right)\right\}dr\right]\notag\\
\leq&E\left[\frac{1}{h}\int_{\bar{t}}^{\bar{t}+h}\left\{e^{-\beta(r-\bar{t})}\left(\beta\varphi
\left(r,X_{\bar{t},\bar{x}}^{\bar u}(r),S_{\bar{t},\bar{s},\ell}^{\bar u}(r)\right)
-\left(\frac{\partial}{\partial t}+\mathcal{L}^{\upsilon}\right)\varphi
\left(r,X_{\bar{t},\bar{x}}^{\bar u}(r),S_{\bar{t},\bar{s},\ell}^{\bar u}(r)\right)\right.\right.\right.\notag\\
&\left.\left.-\phi\left({\bar u}(r)\right)S_{\bar{t},\bar{s},\ell}^{\bar u}(r)\bigg)-{\mathcal Q}V
\left(r,X_{\bar{t},\bar{x}}^{\bar u}(r),S_{\bar{t},\bar{s},\ell}^{\bar u}(r),\ell\right)\right\}dr\Big\vert\:
\hat{\tau}_1\wedge\hat{\tau}_2>\bar{t}+h\right]
P\left[\hat{\tau}_1\wedge\hat{\tau}_2>\bar{t}+h\right]\notag\\
&{} +K\frac{E\left[(\hat{\tau}_1\wedge\hat{\tau}_2-\bar{t})\vert\hat{\tau}_1\wedge\hat{\tau}_2\leq\bar{t}+h\right]}{h}
P\left[\hat{\tau}_1\wedge\hat{\tau}_2\leq\bar{t}+h\right]
\end{align}
for some constant $K$, due to continuity of the function on the left hand side of \eqref{E:Viscosity_1} and the boundedness
of state variable on the time interval $[0,\hat{\tau}_1\wedge\hat{\tau}_2]$.

By definition of $\hat{\tau}_1$, we have
$$
P\left[\hat{\tau}_1\leq\bar{t}+h\right]
=1-P\left[\alpha_{\bar{t},\ell}(r)=\ell,\;r\in(\bar{t},\bar{t}+h]\right]
=-q_{\ell\ell}h.
$$
So as $h\rightarrow 0$, $P\left[\hat{\tau}_1\leq\bar{t}+h\right]$ goes to zero.
By Chebyshev's inequality, we have

\begin{align}\label{E:Supersolution_12}
P\left[\hat{\tau}_2\leq\bar{t}+h\right]
=&P\left[\sup_{r\in[\bar{t},\bar{t}+h]}\left\{\left|X_{\bar{t},\bar{x}}^{\bar u}(r)-\bar{x}\right|^2
+\left|S_{\bar{t},\bar{s},\ell}^{\bar u}(r)-\bar{s}\right|^2\geq\eta^2\right\}\right]\notag\\
\leq&\frac{E\left[\sup_{r\in[\bar{t},\bar{t}+h]}\left|X_{\bar{t},\bar{x}}^{\bar u}(r)-\bar{x}\right|^2\right]
+E\left[\sup_{r\in[\bar{t},\bar{t}+h]}\left|S_{\bar{t},\bar{s},\ell}^{\bar u}(r)-\bar{s}\right|^2\right]}{\eta^2}.
\end{align}

Since each term on the numerator of \eqref{E:Supersolution_12} converges to zero as $h\rightarrow0$ and
$\lim_{h\rightarrow0}P\left[\hat{\tau}_2\leq\bar{t}+h\right]=0$, we have
\begin{align}\label{E:Supersolution_13}
\lim_{h\to0}P\left[\hat{\tau}_1\wedge\hat{\tau}_2\leq\bar{t}+h\right]
\leq\lim_{h\to0}\left(P\left[\hat{\tau}_1\leq\bar{t}+h\right]+P\left[\hat{\tau}_2\leq\bar{t}+h\right]\right)=0.
\end{align}
Let $h\to0$ in \eqref{E:Supersolution_10}. By the mean value theorem and the dominated convergence theorem, we have
$$
\beta\varphi\left(\bar{t},\bar{x},\bar{s}\right)-\left(\frac{\partial}{\partial t}+\mathcal{L}^\upsilon\right)\varphi\left(\bar{t},\bar{x},\bar{s}\right)
-\phi(\upsilon)\bar{s}-{\mathcal Q}V\left(\bar{t},\bar{x},\bar{s},\ell\right)\geq0.
$$

Since ${\bar u}(r)\equiv\upsilon\in U$ is chosen arbitrarily, we take the supremum over $U$ and get
$$
\beta\varphi\left(\bar{t},\bar{x},\bar{s}\right)-\frac{\partial}{\partial t}\varphi\left(\bar{t},\bar{x},\bar{s}\right)
-\sup_{\upsilon\in U}\left\{{\mathcal L}^\upsilon\varphi\left(\bar{t},\bar{x},\bar{s}\right)+\phi(\upsilon)\bar{s}\right\}-{\mathcal Q}
V\left(\bar{t},\bar{x},\bar{s},\ell\right)\geq0.
$$

Therefore, V is a viscosity supersolution of the HJB equation \eqref{E:HJB}.
\end{proof}

For $\ell\in{\mathbb M}$, define the Hamiltonian function ${\mathcal H}$ by
\begin{equation}\label{E:Hamiltonian_1}
{\mathcal H}(t,x,s,p,q,M,\ell):=\sup_{\upsilon\in U}\left\{-\upsilon p+\mu(t,s,\upsilon,\ell)q+\frac{1}{2}\sigma^2(t,s,\upsilon,\ell)M+\phi(\upsilon)s\right\}.
\end{equation}

\begin{lemma}\label{L:Hamiltonian}
For all $\ell\in{\mathbb M}$, the Hamiltonian ${\mathcal H}(t,x,s,p,q,M,\ell)$ is continuous in $(t,x,s,p,q,M)\in[0,T)\times(0,\infty)\times(0,\infty)\times
{\mathbb R}\times{\mathbb R}\times{\mathbb R}$.
\end{lemma}

\begin{proof}

Let the point $(\bar{t},\bar{x},\bar{s},\bar{p},\bar{q},\bar{M})\in[0,T)\times(0,\infty)\times(0,\infty)\times{\mathbb R}^3$ and $B_\eta(\bar{t},\bar{x},\bar{s},\bar{p},\bar{q},\bar{M})$ the ball with the center $(\bar{t},\bar{x},\bar{s},\bar{p},\bar{q},\bar{M})$ and the radius $\eta$, a small constant.
By the definition of the Hamiltonian function, for an arbitrary given $\delta>0$, there exists a  $\bar{\upsilon}\in U$ such that
\begin{align}\label{E:Hamiltonian_3}
{\mathcal H}\left(\bar{t},\bar{x},\bar{s},\bar{p},\bar{q},\bar{M},\ell\right)-\delta\leq-\bar{\upsilon}\bar{p}+\mu\left(\bar{t},\bar{s},\bar{\upsilon},\ell\right)\bar{q}
+\frac{1}{2}\sigma^2\left(\bar{t},\bar{s},\bar{\upsilon},\ell\right)\bar{M}+\phi\left(\bar{\upsilon}\right)\bar{s}.
\end{align}
For any point $(t',x',s',p',q',M')\in B_\eta(\bar{t},\bar{x},\bar{s},\bar{p},\bar{q},\bar{M})$ we also have
\begin{align}
{\mathcal H}\left(t',x',s',p',q',M',\ell\right)&\geq-\bar{\upsilon}p'+\mu\left(t',s',\bar{\upsilon},\ell\right)q'
+\frac{1}{2}\sigma^2\left(t',s',\bar{\upsilon},\ell\right)M'+\phi\left(\bar{\upsilon}\right)s'.\label{E:Hamiltonian_6}
\end{align}
Subtracting \eqref{E:Hamiltonian_6} from \eqref{E:Hamiltonian_3}, we have
\begin{align}\label{E:Hamiltonian_7}
&{\mathcal H}(\bar{t},\bar{x},\bar{s},\bar{p},\bar{q},\bar{M},\ell)-{\mathcal H}(t',x',s',p',q',M',\ell)-\delta\notag\\
&+\frac{1}{2}\sigma^2(\bar{t},\bar{s},\bar{\upsilon},\ell)\left|\bar{M}-M'\right|+\frac{1}{2}\left|\bar{M}+\eta\right|\left|
\sigma^2(\bar{t},\bar{s},\bar{\upsilon},\ell)-\sigma^2(t',s',\bar{\upsilon},\ell)\right|+\left|\phi(\bar{\upsilon})\right|\left|\bar{s}-s'\right|.
\end{align}
Taking the limit inferior and then letting $\delta$ tend to zero in \eqref{E:Hamiltonian_7} we get
\begin{align}\label{E:Hamiltonian_8}
{\mathcal H}(\bar{t},\bar{x},\bar{s},\bar{p},\bar{q},\bar{M},\ell)\leq
\liminf_{\substack{(t',x',s',p',q',M')\\\to(\bar{t},\bar{x},\bar{s},\bar{p},\bar{q},\bar{M})}}
{\mathcal H}(t',x',s',p',q',M',\ell).
\end{align}

Similarly, we can show, using the uniform continuity  of
$\mu(\cdot,\cdot,\cdot,\ell)$ and $\sigma(\cdot,\cdot,\cdot,\ell)$ and the boundedness of the control set $U$, that
\begin{align}\label{E:Hamiltonian_10}
\limsup_{\substack{(t',x',s',p',q',M')\\\to(\bar{t},\bar{x},\bar{s},\bar{p},\bar{q},\bar{M})}}
{\mathcal H}(t',x',s',p',q',M',\ell)\leq
{\mathcal H}(\bar{t},\bar{x},\bar{s},\bar{p},\bar{q},\bar{M},\ell).
\end{align}
(\ref{E:Hamiltonian_8}) and (\ref{E:Hamiltonian_10}) imply that
the Hamiltonian ${\mathcal H}(t,x,s,p,q,M,\ell)$ is continuous in $(t,x,s,p,q,M)$.

\end{proof}

For $\varphi\in C^{1,1,2}$ Theorem~\ref{T:Continuity} and  Lemma~\ref{L:Hamiltonian} imply that
the mapping
\begin{align}\label{E:Hamiltonian_13}
(t,x,s)\mapsto\beta\varphi(t,x,s)-\frac{\partial\varphi}{\partial t}(t,x,s)-\sup_{\upsilon\in U}\left\{\mathcal{L}^\upsilon
\varphi(t,x,s)+\phi(\upsilon)s\right\}-{\mathcal Q}V(t,x,s,\ell)
\end{align}
is continuous.

\begin{theorem} \label{T:Subsolution}
For each $\ell\in{\mathbb M}$, the value function $V=\{V(t,x,s,\ell)\}_{\ell\in{\mathbb M}}$ is a viscosity subsolution of the HJB equation \eqref{E:HJB}.
\end{theorem}

\begin{proof}
Assume, for contradiction, that $V$ is not a viscosity subsolution. Then there exists $\ell\in{\mathbb M}$, $(\bar{t},\bar{x},\bar{s})\in[0,T)
\times(0,\infty)\times(0,\infty)$ and a test function $\varphi(t,x,s)\in C^{1,1,2}([0,T)\times(0,\infty)\times(0,\infty))$ such that
\begin{equation}\label{E:Subsolution_1}
\beta\varphi({\bar t},{\bar x},{\bar s})-\frac{\partial\varphi}{\partial t}({\bar t},{\bar x},{\bar s})-\sup_{\upsilon\in U}
\left\{{\mathcal L}^\upsilon\varphi({\bar t},{\bar x},{\bar s})+\phi(\upsilon){\bar s}\right\}-{\mathcal Q}V({\bar t},{\bar x},{\bar s},\ell)>0,
\end{equation}
where $V(t,x,s,\ell)-\varphi(t,x,s)$ attains its maximum at $(\bar{t},\bar{x},\bar{s})$. Without loss of generality, assume that
$V(\bar{t},\bar{x},\bar{s},\ell)-\varphi(\bar{t},\bar{x},\bar{s})=0$.

By the continuity of the mapping in \eqref{E:Hamiltonian_13}, for $\delta>0$, there exists $\eta>0$ such that
\begin{equation}\label{E:Subsolution_2}
\beta\varphi\left(t,x,s\right)-\frac{\partial\varphi}{\partial t}\left(t,x,s\right)-\sup_{\upsilon\in U}
\left\{\mathcal{L}^\upsilon\varphi\left(t,x,s\right)+\phi(\upsilon)s\right\}-{\mathcal Q}V(t,x,s,\ell)\geq\delta
\end{equation}
for all $(t,x,s)\in B_\eta(\bar{t},\bar{x},\bar{s})$. Let $\eta$ be small enough such that $B_\eta(\bar{t},\bar{x},\bar{s})\subset
[0,T)\times(0,\infty)\times(0,\infty)$.

Let $h>0$ be small enough such that $({\bar t},{\bar t}+h)\subset[0,T)$.
By dynamic programming principle, there exists a control process ${\bar u}\in\mathcal U$ such that
\begin{equation}\label{E:Subsolution_3}
V\left(\bar{t},\bar{x},\bar{s},\ell\right)-\frac{\delta}{2}h\leq E\left[\int_{\bar{t}}^{\tau}e^{-\beta(r-\bar{t})}\phi\left({\bar u}(r)\right)
S_{\bar{t},\bar{s},\ell}^{\bar u}(r)dr+e^{-\beta(\tau-\bar{t})}V\left(\tau,X_{\bar{t},\bar{x}}^{\bar u}(\tau),S_{\bar{t},\bar{s},\ell}^{\bar u}(\tau),
\alpha_{\bar{t},\ell}(\tau)\right)\right],
\end{equation}
where $\tau\geq t$ is any stopping time.
Let $\hat{\tau}_1$ be the first jump time of $\alpha_{\bar{t},\ell}(\cdot)$ and define the exit time
$$
\hat{\tau}_3:=\inf\left\{r\geq\bar{t}:\left(r,X_{\bar{t},\bar{x}}^{\bar u}(r),S_{\bar{t},\bar{s},\ell}^{\bar u}(r)\right)\not\in
B_\eta(\bar{t},\bar{x},\bar{s})\right\}.
$$
Let $\tau:=(\bar{t}+h)\wedge\hat{\tau}_1\wedge\hat{\tau}_3$ and
define a function $\psi(t,x,s,i)$
as in \eqref{E:Supersolution_3}. We have
$$
\psi\left(\tau,X_{\bar{t},\bar{x}}^{\bar u}(\tau),S_{\bar{t},\bar{s},\ell}^{\bar u}(\tau),\alpha_{\bar{t},\ell}(\tau)\right)
\geq V\left(\tau,X_{\bar{t},\bar{x}}^{\bar u}(\tau),S_{\bar{t},\bar{s},\ell}^{\bar u}(\tau),\alpha_{\bar{t},\ell}(\tau)\right)
$$
and
\begin{equation}\label{E:Subsolution_6}
{\mathcal Q}\psi\left(r,X_{\bar{t},\bar{x}}^{\bar u}(r),S_{\bar{t},\bar{s},\ell}^{\bar u}(r),\ell\right)
\leq{\mathcal Q}V\left(r,X_{\bar{t},\bar{x}}^{\bar u}(r),S_{\bar{t},\bar{s},\ell}^{\bar u}(r),\ell\right).
\end{equation}

So equation \eqref{E:Subsolution_3} turns into
\begin{equation}\label{E:Subsolution_7}
\varphi\left(\bar{t},\bar{x},\bar{s}\right)-\frac{\delta}{2}h\leq E\left[\int_{\bar{t}}^{\tau}e^{-\beta(r-\bar{t})}
\phi\left({\bar u}(r)\right)S_{\bar{t},\bar{s},\ell}^{\bar u}(r)dr+e^{-\beta(\tau-\bar{t})}\psi
\left(\tau,X_{\bar{t},\bar{x}}^{\bar u}(\tau),S_{\bar{t},\bar{s},\ell}^{\bar u}(\tau),\alpha_{\bar{t},\ell}(\tau)\right)\right].
\end{equation}

Combining \eqref{E:Supersolution_8}, \eqref{E:Subsolution_6} and \eqref{E:Subsolution_7}, we divide both sides of the equation
by $h$,
\begin{align}\label{E:Subsolution_8}
0\geq-\frac{\delta}{2}&+E\left[\frac{1}{h}\int_{\bar{t}}^{\tau}\left\{e^{-\beta(r-\bar{t})}\left[\beta\varphi
\left(r,X_{\bar{t},\bar{x}}^{\bar u}(r),S_{\bar{t},\bar{s},\ell}^{\bar u}(r)\right)
-\left(\frac{\partial}{\partial t}+\mathcal{L}^{{\bar u}(r)}\right)\right.\right.\right.\notag\\
&\left.\left.\varphi\left(r,X_{\bar{t},\bar{x}}^{\bar u}(r),S_{\bar{t},\bar{s},\ell}^{\bar u}(r)\right)
-\phi\left({\bar u}(r)\right)S_{\bar{t},\bar{s},\ell}^{\bar u}(r)\bigg]-{\mathcal Q}V
\left(r,X_{\bar{t},\bar{x}}^{\bar u}(r),S_{\bar{t},\bar{s},\ell}^{\bar u}(r),\ell\right)\right\}dr\right].
\end{align}

Substituting \eqref{E:Subsolution_2} into \eqref{E:Subsolution_8}, we have
\begin{equation}\label{E:Subsolution_9}
0\geq-\frac{\delta}{2}+\frac{\delta}{h}E[\tau-\bar{t}].
\end{equation}
By \eqref{E:Supersolution_13}, we have
\begin{align}
1\geq&\frac{1}{h}E\left[\tau-\bar{t}\right]
\geq \frac{1}{h} E\left[h\;\vert\hat{\tau}_1\wedge\hat{\tau}_3>\bar{t}
+h\right]P\left[\hat{\tau}_1\wedge\hat{\tau}_2>\bar{t}+h\right]
=P\left[\hat{\tau}_1\wedge\hat{\tau}_3>\bar{t}+h\right]\to1\notag
\end{align}
as $h\to0$, which implies that
$$
\lim_{h\to0}\frac{1}{h}E\left[\tau-\bar{t}\right]=1.
$$

Letting $h\to0$ in \eqref{E:Subsolution_9}, we get $\delta/2\leq0$,  a contradiction. The inequality in \eqref{E:Subsolution_1} therefore holds, which completes the proof.

\end{proof}

Since the value function $V$ is both a viscosity subsolution and a viscosity supersolution,
we conclude that it is a viscosity solution of the HJB equation \eqref{E:HJB}. We have proved Theorem \ref{T:solution}.

\section{Proof of Theorem~\ref{T:Comparison}} \label{S:Uniqueness}
In this section vectors $(t,x,s)$ and $(r,y,v)$ and their specific values such as $(\bar t,\bar x,\bar s)$ appear many times. To simplify the expressions we denote by $\mathbf x=(t,x,s)$ and $\mathbf y=(r,y,v)$. Their specific values are defined similarly, for example,
$\bar{\mathbf x}=(\bar t,\bar x,\bar s)$.

To prove the uniqueness, we need an alternative definition of viscosity solution in terms of superjets and subjets. The second-order superjet of an upper-semicontinuous  function $U$ at a point $\bar{\mathbf x} \in \Sigma:=[0,T)\times(0,\infty)\times(0,\infty)$, denoted by  ${\mathcal P}^{2,+}U(\bar{\mathbf x})$, is defined as a set of elements $(\bar{b},\bar{p},\bar{q},\bar{M})\in{\mathbb R}\times{\mathbb R}\times{\mathbb R}\times{\mathbb R}$ such that
\begin{equation} \label{superjet}
U(\mathbf x)\leq U(\bar{\mathbf x})+(\bar{b}, \bar p, \bar{q}) \cdot (\mathbf x-\bar{\mathbf x})+\frac{1}{2}\bar{M}(s-\bar{s})^2+e(\mathbf x-\bar{\mathbf x}),
\end{equation}
where $e(\mathbf x-\bar{\mathbf x})=o(\left|t-\bar{t}\right|+\left|x-\bar{x}\right|+\left|s-\bar{s}\right|^2)$ is a higher order error term.
The limiting superjet $\overline{\mathcal P}^{2,+}U(\mathbf x)$ is the set of elements $(b,p,q,M)\in{\mathbb R}^4$ for which there exists
a sequence $({\mathbf x}_\epsilon)$ in $\Sigma$ and
$(b_\epsilon,p_\epsilon,q_\epsilon,M_\epsilon)\in{\mathcal P}^{2,+}U({\mathbf x}_\epsilon)$ such that
$({\mathbf x}_\epsilon,U({\mathbf x}_\epsilon),b_\epsilon,p_\epsilon,q_\epsilon,M_\epsilon)\to
(\mathbf x,U(\mathbf x),b,p,q,M)$.

The second-order subjet of a lower-semicontinuous function $V$ at a point
$\bar{\mathbf x} \in \Sigma$, denoted by  ${\mathcal P}^{2,-}V(\bar{\mathbf x} )$, is defined as in (\ref{superjet}) with a greater than or equal ($\geq$) inequality. The set $\overline{\mathcal P}^{2,-}V(\mathbf x)$ is defined similarly.

Note that since $x$ is a state variable superjets and subjects should normally also have  second order terms with respect to $x$. However, since the HJB equation \eqref{E:HJB} only involves the first order derivative  of the value function with respect to $x$, the second order expansion  in $x$ is not needed.

Assume that $U$ is upper-semicontinuous and $\varphi\in C^{1,1,2}(\Sigma)$. Then $\bar{\mathbf x}\in \Sigma$ is a maximum point of $U-\varphi$ if and only if $(D_{{\mathbf x}}\varphi(\bar{\mathbf x}),D^2_s\varphi
(\bar{\mathbf x}))\in{\mathcal P}^{2,+}U(\bar{\mathbf x})$, where
$D_{{\mathbf x}}\varphi(\bar{\mathbf x}) =(D_t\varphi(\bar{\mathbf x}),D_x\varphi(\bar{\mathbf x}),D_s\varphi(\bar{\mathbf x}))$. Similar conclusion holds  for the minimum point and the subjet.

\begin{lemma}\label{L:Viscosity_Alt} (\cite[Theorem 8.3]{CIL92})
An $m$-tuple $V=\{V(\cdot,\cdot,\cdot,\ell)\}_{\ell\in{\mathbb M}}$ of continuous functions on $\Sigma$ is a viscosity
subsolution (resp. supersolution) of the HJB equation \eqref{E:HJB} if and only if for
${\mathbf x}\in\Sigma$ such that $(b,p,q,M)\in\overline{\mathcal P}^{2,+}V({\mathbf x},\ell)$
(resp. $\in\overline{\mathcal P}^{2,-}V({\mathbf x},\ell)$) for any fixed $\ell\in{\mathbb M}$, we have
\begin{align*}
\beta V({\mathbf x},\ell)-b-&{\mathcal H}({\mathbf x},p,q,M,\ell)-
{\mathcal Q}V({\mathbf x},\ell)\leq0\quad
(\mbox{resp.~} \geq 0),
\end{align*}
where ${\mathcal H}({\mathbf x},p,q,M,\ell)$ is the Hamiltonian define in (\ref{E:Hamiltonian_1}). The $m$-tuple $V$ is a viscosity solution if it is both a viscosity subsolution and a viscosity supersolution.
\end{lemma}


The uniform polynomial growth condition  for $U$ and $V$ implies that there exists a constant $p>1$ such that, for each $\ell\in{\mathbb M}$
$$
\sup_{[0,T]\times(0,\infty)\times(0,\infty)}\frac{\left|U(\mathbf x,\ell)\right|+\left|V(\mathbf x,\ell)\right|}{1+\left|x\right|^p+\left|s\right|^p}<\infty.
$$

Define functions $\theta(x,s):=(1+\left|x\right|^{2p}+\left|s\right|^{2p})$ and $\kappa(t,x,s):=e^{-\gamma t}\theta(x,s)$ for $\gamma>0$.
Due to the linear growth condition \eqref{E:ModelSetup5} and the boundedness of  set $U$, there exists a
positive constant $c$ such that, for all $\ell\in{\mathbb M}$,
\begin{align}
&\beta\kappa-\frac{\partial\kappa}{\partial t}-\sup_{\upsilon\in U}\left\{{\mathcal L}^\upsilon\kappa\right\}\notag\\
=&\beta\kappa-\frac{\partial\kappa}{\partial t}-\sup_{\upsilon\in U}\left\{-\upsilon D_x\kappa+\mu(t,s,\upsilon,\ell)D_s\kappa
+\frac{1}{2}\sigma^2(t,s,\upsilon,\ell)D_s^2\kappa+{\mathcal Q}\kappa\right\}\notag\\
=&e^{-\gamma t}\left[(\beta+\gamma)\theta-\sup_{\upsilon\in U}\left\{-\upsilon D_x\theta+\mu(t,s,\upsilon,\ell)D_s\theta+\frac{1}{2}\sigma^2(t,s,\upsilon,\ell)D_s^2\theta
+{\mathcal Q}\theta\right\}\right]\notag\\
\geq&e^{-\gamma t}(\beta+\gamma-c)\theta,\notag
\end{align}
which is nonnegative as long as we choose the constant $\gamma$ large enough such that $(\beta+\gamma-c)>0$. Therefore, for any $\epsilon>0$, $\widetilde{V}^\epsilon(\mathbf x,\ell)
:=V(\mathbf x,\ell)+\epsilon\kappa(\mathbf x)$ is a supersolution to the HJB equation \eqref{E:HJB}.
To check this, let $\varphi(\mathbf x,\ell)$ be the test function for $\widetilde{V}^\epsilon(\mathbf x,\ell)$. So
$\varphi(\mathbf x,\ell)-\epsilon\kappa(\mathbf x)$ is the test function for the supersolution $V(\mathbf x,\ell)$.
We have
\begin{align}
&\beta\varphi-\frac{\partial\varphi}{\partial t}-\sup_{\upsilon\in U}\left\{{\mathcal L}^\upsilon\varphi+\phi(\upsilon)s\right\}\notag\\
\geq&\beta(\varphi-\epsilon\kappa)-\frac{\partial}{\partial t}(\varphi-\epsilon\kappa)-\sup_{\upsilon\in U}
\left\{{\mathcal L}^\upsilon(\varphi-\epsilon\kappa)+\phi(\upsilon)s\right\}+\epsilon\left(\beta\kappa-\frac{\partial\kappa}{\partial t}-\sup_{\upsilon\in U}\left\{{\mathcal L}^\upsilon\kappa\right\}\right)\notag\\
\geq&0.\notag
\end{align}
By the polynomial growth condition of $U$, $V$ and the definition of $\kappa$, we have
$$
\lim_{x,s\to\infty}\sup_{t\in[0,T]}(U-\widetilde{V}^\epsilon)(\mathbf x,\ell)=-\infty
$$
for all $\epsilon>0$.
We can assume that the maximum of $(U-V)(\mathbf x,\ell)$ over $\ell\in{\mathbb M}$ and ${\mathbf x}\in[0,T]\times(0,\infty)\times(0,\infty)$ is
attained (up to a penalization) at $\ell\in{\mathbb M}$ and ${\mathbf x}\in\Sigma_1:=[0,T]\times O_1\times O_2$ for some compact set $O_1\subset(0,\infty)$
and $O_2\subset(0,\infty)$.  Let ${\mathcal M}$ denote this maximum.

Suppose, for contradiction, that there exists $\ell\in{\mathbb M}$
and ${\mathbf x}\in\Sigma$ such that $U(\mathbf x,\ell)>V(\mathbf x,\ell)$. We have
\begin{align}\label{E:Comparison_5}
{\mathcal M}:=&\max_{i\in{\mathbb M}}\sup_{[0,T]\times(0,\infty)^2}(U-V)(\mathbf x,i)=\max_{i\in{\mathbb M},{\mathbf x}\in\Sigma_1}(U-V)(\mathbf x,i)>0.
\end{align}

For any $\epsilon>0$, define a function $\Psi^\epsilon$ by
$$
\Psi^\epsilon({\mathbf x},{\mathbf y},\ell):=U({\mathbf x},\ell)-V({\mathbf y},\ell)-\psi^\epsilon({\mathbf x},{\mathbf y}),
$$
where $\psi^\epsilon$ is defined by
\begin{equation}\label{E:Comparison_7}
\psi^\epsilon({\mathbf x},{\mathbf y}):=\frac{1}{2\epsilon}\left|{\mathbf x}-{\mathbf y}\right|^2.
\end{equation}

For each $\ell\in{\mathbb M}$, $\Psi^\epsilon(\cdot,\cdot,\ell)$ is continuous. Hence its maximum, denoted by ${\mathcal M}_\ell^\epsilon$,  over the compact set
$\Sigma_1\times\Sigma_1$ can be attained at
$({\mathbf x}_\ell^\epsilon,{\mathbf y}_\ell^\epsilon)$. Assume that the maximum ${\mathcal M}^\epsilon
:=\max_{\ell\in{\mathbb M}}{\mathcal M}_\ell^\epsilon$ is attained at $\ell^\epsilon\in{\mathbb M}$ and
$({\mathbf x}_{\ell^\epsilon}^\epsilon,{\mathbf y}_{\ell^\epsilon}^\epsilon)$. We have
\begin{align}\label{E:Comparison_8}
{\mathcal M}\leq{\mathcal M}^\epsilon=&\Psi^\epsilon({\mathbf x}_{\ell^\epsilon}^\epsilon,{\mathbf y}_{\ell^\epsilon}^\epsilon,\ell^\epsilon)
\leq U({\mathbf x}_{\ell^\epsilon}^\epsilon,\ell^\epsilon)-V({\mathbf y}_{\ell^\epsilon}^\epsilon,\ell^\epsilon).
\end{align}

As $\epsilon\to0$, the bounded sequence $({\mathbf x}_{\ell^\epsilon}^\epsilon,{\mathbf y}_{\ell^\epsilon}^\epsilon)$ converges,
up to a subsequence, to a limit $(\bar{\mathbf x},\bar{\mathbf y})\in\Sigma_1\times\Sigma_1$.
By assumption, ${\mathbb M}$ is finite. For each $\ell\in{\mathbb M}$, the sequence $({\mathbf x}_{\ell}^\epsilon,{\mathbf y}_{\ell}^\epsilon)$
converges, up to a subsequence, to its limit, respectively. Therefore, for $\epsilon$ small
enough, $\ell^\epsilon=\bar{\ell}$ for $\bar{\ell}\in{\mathbb M}$.

Since $\{U(\cdot,\ell)\}_{\ell\in{\mathbb M}}$ and $\{V(\cdot,\ell)\}_{\ell\in{\mathbb M}}$ are continuous and ${\mathbb M}$ is a finite set,
$U({\mathbf x}_{\ell^\epsilon}^\epsilon,\ell^\epsilon)-V({\mathbf y}_{\ell^\epsilon}^\epsilon,\ell^\epsilon)$ is bounded for all $\epsilon>0$.
From \eqref{E:Comparison_8}, $\psi^\epsilon({\mathbf x}_{\ell^\epsilon}^\epsilon,{\mathbf y}_{\ell^\epsilon}^\epsilon)$ is also bounded, which implies that
\begin{equation} \label{eqn1}
\lim_{\epsilon\to0}({\mathbf x}_{\ell^\epsilon}^\epsilon,{\mathbf y}_{\ell^\epsilon}^\epsilon) = (\bar{\mathbf x},\bar{\mathbf x}),\quad
\lim_{\epsilon\to0}{\mathcal M}^\epsilon={\mathcal M}=(U-V)(\bar{\mathbf x},\bar{\ell}), \quad
\lim_{\epsilon\to0}\psi^\epsilon({\mathbf x}_{\ell^\epsilon}^\epsilon,{\mathbf y}_{\ell^\epsilon}^\epsilon)=0.
\end{equation}

By applying Ishii's Lemma (see \cite[Lemma 4.4.6, Remark 4.4.9]{P10}) to function $\Psi^{\epsilon}$ at its maximum point $({\mathbf x}_{\ell^\epsilon}^\epsilon,{\mathbf y}_{\ell^\epsilon}^\epsilon)$ with $\ell=\ell^\epsilon$,
we can find $M^\epsilon,N^\epsilon\in{\mathbb R}$ such that
\begin{align*}
\left(\frac{1}{\epsilon}\left({\mathbf x}_{\ell^\epsilon}^\epsilon-{\mathbf y}_{\ell^\epsilon}^\epsilon\right),M^\epsilon\right)\in\overline{\mathcal P}^{2,+}
U({\mathbf x}_{\ell^\epsilon}^\epsilon, \ell^\epsilon),\quad
\left(\frac{1}{\epsilon}\left({\mathbf x}_{\ell^\epsilon}^\epsilon-{\mathbf y}_{\ell^\epsilon}^\epsilon\right),N^\epsilon\right)\in\overline{\mathcal P}^{2,-}
V({\mathbf y}_{\ell^\epsilon}^\epsilon, \ell^\epsilon)
\end{align*}
and, for any $c,d\in{\mathbb R}$,
\begin{align}\label{E:Comparison_15}
c^2M^\epsilon-d^2N^\epsilon\leq\frac{3}{\epsilon}(c-d)^2.
\end{align}

Denote by
$$
(\eta^\epsilon_1, \eta^\epsilon_2,\eta^\epsilon_3)
:= \frac{1}{\epsilon}\left({\mathbf x}_{\ell^\epsilon}^\epsilon-{\mathbf y}_{\ell^\epsilon}^\epsilon\right)
=\left(\frac{1}{\epsilon}(t_{\ell^\epsilon}^\epsilon-
r_{\ell^\epsilon}^\epsilon),
\frac{1}{\epsilon}(x_{\ell^\epsilon}^\epsilon-y_{\ell^\epsilon}^\epsilon),
\frac{1}{\epsilon}(s_{\ell^\epsilon}^\epsilon-v_{\ell^\epsilon}^\epsilon)
\right).
$$
Since $U$ is a viscosity subsolution and $V$ a supersolution, by the definition of viscosity solutions
in terms of superjets and subjets, we have
\begin{align}
\beta U({\mathbf x}_{\ell^\epsilon}^\epsilon,\ell^\epsilon)&-\eta^\epsilon_1-{\mathcal Q}U({\mathbf x}_{\ell^\epsilon}^\epsilon,\ell^\epsilon)-{\mathcal H}\left({\mathbf x}_{\ell^\epsilon}^\epsilon,\eta^\epsilon_2,
\eta^\epsilon_3,M^\epsilon,\ell^\epsilon\right)\leq0\label{E:Comparison_16}\\
\beta V({\mathbf y}_{\ell^\epsilon}^\epsilon,\ell^\epsilon)&-\eta^\epsilon_1-{\mathcal Q}V({\mathbf y}_{\ell^\epsilon}^\epsilon,\ell^\epsilon)-{\mathcal H}\left({\mathbf y}_{\ell^\epsilon}^\epsilon,\eta^\epsilon_2,
\eta^\epsilon_3,N^\epsilon,\ell^\epsilon\right)\geq0,\label{E:Comparison_17}
\end{align}
where the Hamiltonian ${\mathcal H}$ is defined in \eqref{E:Hamiltonian_1}.
By the definition of operator ${\mathcal Q}$ we have
\begin{align}\label{E:Comparison_19}
&{\mathcal Q}\left(U({\mathbf x}_{\ell^\epsilon}^\epsilon,\ell^\epsilon)-
V({\mathbf y}_{\ell^\epsilon}^\epsilon,\ell^\epsilon)\right)\notag\\
=&\sum_{j\neq\ell^\epsilon}q_{\ell^\epsilon j}\left[\left(U({\mathbf x}_{\ell^\epsilon}^\epsilon,j)
-V({\mathbf y}_{\ell^\epsilon}^\epsilon,j)\right)
-\left(U({\mathbf x}_{\ell^\epsilon}^\epsilon,\ell^\epsilon)
-V({\mathbf y}_{\ell^\epsilon}^\epsilon,\ell^\epsilon)\right)\right]\notag\\
=&\sum_{j\neq\ell^\epsilon}q_{\ell^\epsilon j}\left[\Psi^\epsilon({\mathbf x}_{\ell^\epsilon}^\epsilon,{\mathbf y}_{\ell^\epsilon}^\epsilon,j)
-\Psi^\epsilon({\mathbf x}_{\ell^\epsilon}^\epsilon,{\mathbf y}_{\ell^\epsilon}^\epsilon,\ell^\epsilon)\right]\notag\\
\leq&0.
\end{align}
The last line is from the fact that
$\Psi^\epsilon({\mathbf x}_{\ell^\epsilon}^\epsilon,{\mathbf y}_{\ell^\epsilon}^\epsilon,\ell^\epsilon)$ is the maximum of
$\Psi^\epsilon({\mathbf x},{\mathbf y},\ell)$ over $\ell\in{\mathbb M}$ and $({\mathbf x},{\mathbf y})\in\Sigma_1\times\Sigma_1$.

Subtracting \eqref{E:Comparison_17} from \eqref{E:Comparison_16} and rearranging, also noting (\ref{E:Comparison_19}), we have
\begin{align}\label{E:Comparison_18}
&\beta\left(U({\mathbf x}_{\ell^\epsilon}^\epsilon,\ell^\epsilon)-V({\mathbf y}_{\ell^\epsilon}^\epsilon,\ell^\epsilon)\right)
\leq {\mathcal H}\left({\mathbf x}_{\ell^\epsilon}^\epsilon,
\eta^\epsilon_2,
\eta^\epsilon_3,M^\epsilon,\ell^\epsilon\right)-{\mathcal H}\left({\mathbf y}_{\ell^\epsilon}^\epsilon,
\eta^\epsilon_2,
\eta^\epsilon_3,N^\epsilon,\ell^\epsilon\right).
\end{align}

By the definition of the  Hamiltonian function, for any $\delta>0$, there exists a  $\upsilon^\delta\in U$ such that
\begin{equation}\label{E:Comparison_20}
{\mathcal H}\left({\mathbf x}_{\ell^\epsilon}^\epsilon,
\eta^\epsilon_2,
\eta^\epsilon_3,M^\epsilon,\ell^\epsilon\right)-\delta
\leq-\upsilon^\delta\eta^\epsilon_2+\mu\left(t_{\ell^\epsilon}^\epsilon,
s_{\ell^\epsilon}^\epsilon,\upsilon^\delta,\ell^\epsilon\right)\eta^\epsilon_3
+\frac{1}{2}\sigma^2\left(t_{\ell^\epsilon}^\epsilon,s_{\ell^\epsilon}^\epsilon,\upsilon^\delta,\ell^\epsilon\right)M^\epsilon
+\phi(\upsilon^\delta)s_{\ell^\epsilon}^\epsilon.
\end{equation}
We also have
\begin{equation}\label{E:Comparison_21}
{\mathcal H}\left({\mathbf y}_{\ell^\epsilon}^\epsilon,
\eta^\epsilon_2,
\eta^\epsilon_3,N^\epsilon,\ell^\epsilon\right)
\geq-\upsilon^\delta\eta^\epsilon_2+\mu\left(r_{\ell^\epsilon}^\epsilon,
v_{\ell^\epsilon}^\epsilon,\upsilon^\delta,\ell^\epsilon\right)\eta^\epsilon_3
+\frac{1}{2}\sigma^2\left(s_{\ell^\epsilon}^\epsilon,v_{\ell^\epsilon}^\epsilon,\upsilon^\delta,\ell^\epsilon\right)N^\epsilon
+\phi(\upsilon^\delta)v_{\ell^\epsilon}^\epsilon.
\end{equation}
Subtracting \eqref{E:Comparison_21} from \eqref{E:Comparison_20}, we get
\begin{align}\label{E:Comparison_22}
&{\mathcal H}\left({\mathbf x}_{\ell^\epsilon}^\epsilon,
\eta^\epsilon_2,
\eta^\epsilon_3,M^\epsilon,\ell^\epsilon\right)-{\mathcal H}\left({\mathbf y}_{\ell^\epsilon}^\epsilon,
\eta^\epsilon_2,
\eta^\epsilon_3,N^\epsilon,\ell^\epsilon\right)-\delta\notag\\
\leq&\eta^\epsilon_3\left[
\mu\left(t_{\ell^\epsilon}^\epsilon,s_{\ell^\epsilon}^\epsilon,\upsilon^\delta,\ell^\epsilon\right)
-\mu\left(r_{\ell^\epsilon}^\epsilon,v_{\ell^\epsilon}^\epsilon,\upsilon^\delta,\ell^\epsilon\right)\right]\notag\\
&+\frac{3}{2\epsilon}(\sigma\left(t_{\ell^\epsilon}^\epsilon,s_{\ell^\epsilon}^\epsilon,\upsilon^\delta,\ell^\epsilon\right)
-\sigma\left(r_{\ell^\epsilon}^\epsilon,v_{\ell^\epsilon}^\epsilon,\upsilon^\delta,\ell^\epsilon\right))^2
+\phi(\upsilon^\delta)\left(s_{\ell^\epsilon}^\epsilon-v_{\ell^\epsilon}^\epsilon\right).
\end{align}
Here we have used  \eqref{E:Comparison_15}.

By Assumption~\ref{ass1} on $\mu$ and $\sigma$, (\ref{eqn1}) and the boundedness of $\phi(\upsilon^\delta)$, the right side of \eqref{E:Comparison_22} tends to 0
as $\epsilon\rightarrow0$. Since $\delta>0$ is chosen arbitrarily, we have
\begin{align}\label{E:Comparison_23}
\limsup_{\epsilon\to0}&\left\{{\mathcal H}\left({\mathbf x}_{\ell^\epsilon}^\epsilon,
\eta^\epsilon_2,
\eta^\epsilon_3,M^\epsilon,\ell^\epsilon\right)\right.\left.-{\mathcal H}\left({\mathbf y}_{\ell^\epsilon}^\epsilon,
\eta^\epsilon_2,
\eta^\epsilon_3,N^\epsilon,\ell^\epsilon\right)\right\}\leq0.
\end{align}
Combining (\ref{eqn1}), \eqref{E:Comparison_18} and \eqref{E:Comparison_23}, we have
$$
\beta\left(U\left(\bar{\mathbf x},\bar{\ell}\right)-V\left(\bar{\mathbf x},\bar{\ell}\right)\right)\leq0,
$$
which contradicts  \eqref{E:Comparison_5}. Therefore
$U\leq V$ on $[0,T)\times(0,\infty)\times(0,\infty)\times{\mathbb M}$.
We have proved Theorem~\ref{T:Comparison}.

\bigskip
\noindent{\bf Acknowledgement}. The authors thank two anonymous referees for their suggestions and comments that have helped to improve the paper.

\end{document}